\documentclass[final]{article}

\usepackage[numbers,sort&compress]{natbib}
\usepackage{neurips-template/neurips_2026}
\usepackage[T1]{fontenc}
\usepackage{lmodern}
\usepackage{amsmath,amssymb,amsthm}
\usepackage{mathtools}
\usepackage{bm}

\usepackage{graphicx}
\usepackage{float}
\usepackage[labelfont=bf,font=small]{caption}
\usepackage{subcaption}

\usepackage{booktabs}
\usepackage{tabularx}
\usepackage{multirow}
\usepackage{array}
\usepackage{colortbl}
\usepackage{makecell}
\usepackage{footnote}

\usepackage[table]{xcolor}

\definecolor{eqclo}{RGB}{234,243,222}   
\definecolor{eqcmd}{RGB}{192,221,151}   
\definecolor{eqchi}{RGB}{250,238,218}   
\definecolor{eqcvh}{RGB}{250,199,117}   
\definecolor{eqcex}{RGB}{240,153,123}   
\definecolor{indet}{RGB}{216, 90,  48}  

\definecolor{CapicuRed} {RGB}{ 175,0, 0}
\definecolor{CapicuSlate}{RGB}{ 44, 44, 42}
\definecolor{CapicuMuted}{RGB}{100,100, 96}
\definecolor{CapicuRule} {RGB}{200,200,196}

\usepackage[ruled,linesnumbered,vlined]{algorithm2e}

\usepackage[colorlinks=true,
            linkcolor=blue!70!black,
            citecolor=blue!70!black,
            urlcolor=CapicuRed]{hyperref}
\usepackage{url}
\usepackage{fontawesome5}

\usepackage{enumitem}
\usepackage{mdframed}
\usepackage{tcolorbox}          

\newcommand{\rpn}{\mathrm{RP}}

\newcommand{\grid}{\mathcal{G}}

\newcommand{\ort}{\textsc{ONNX Runtime}}

\theoremstyle{definition}
\newtheorem{cdef}{Definition}
\newtheorem{cprop}{Proposition}

\newcommand{\PaperTitle}{%
  INT8 Quantization Makes ARM Edge Inference Dispatch-Invariant
}

\newcommand{\PaperAuthors}{%
  Sebasti\'{a}n A. Cruz Romero\textsuperscript{1}, Shenied E. Maldonado Guerra\textsuperscript{1}
}

\newcommand{\PaperAffils}{%
  \textsuperscript{1}Capic\'{u} Technologies
  \\[2pt]
}

\newcommand{\PaperVersion}{\url{https://capicu.ai}}
\newcommand{\PaperAbstractText}{
On x86, kernel dispatch fragments the outputs of the same neural network into many equivalence classes across hardware. We ask whether the same fragmentation governs ARM edge inference, where most edge ML actually runs. Across four Raspberry Pi devices spanning Cortex-A53, A72, and A76 under ONNX Runtime CPU, microarchitecture is not observable in the outputs of a fixed FP32 CNN. Holding hardware constant at Cortex-A76 and switching only the execution provider, FP32 outputs disagree on every CIFAR-10 image with a mean remaining precision of 14.97 of 23 mantissa bits. INT8 QDQ post-training quantization collapses both axes to a single equivalence class. We trace this to a structural property of QDQ graphs that we call H1+H2: discrete-grid inputs make any Conv dispatch-deterministic (H1) and QuantizeLinear at every layer boundary preserves that precondition (H2). H1+H2 predicts that bit-exact agreement should extend to production CNNs under runtimes that confirmably exercise different ARM microkernels. We verify this on MobileNetV2 and ResNet50V2 under TensorFlow Lite with XNNPACK, where timing evidence confirms SDOT dispatch on A76 and NEON multiply-accumulate on A72 yet every intermediate INT32 accumulator and every final output is byte-identical across 500 ImageNet images per model. We then identify the specific x86 mechanism that breaks the same invariant on x86, namely \texttt{PMADDUBSW} saturating INT16 intermediates, which has no ARM analogue. The Schl\"ogl et al. divergence phenomenon is delineated rather than contradicted. For practitioners deploying quantized CNNs across heterogeneous ARM fleets, the operational consequence is direct. INT8 inference is the reproducible mode and the relevant behavioral variation axis is precision, not microarchitecture.
}

\newcommand{\CapicuHeader}{%
  \thispagestyle{empty}%
  \setlength{\parindent}{0pt}%
  \newlength{\logowidth}\setlength{\logowidth}{0.18\linewidth}%
  \begin{minipage}[c]{\logowidth}%
    \includegraphics[width=\linewidth]{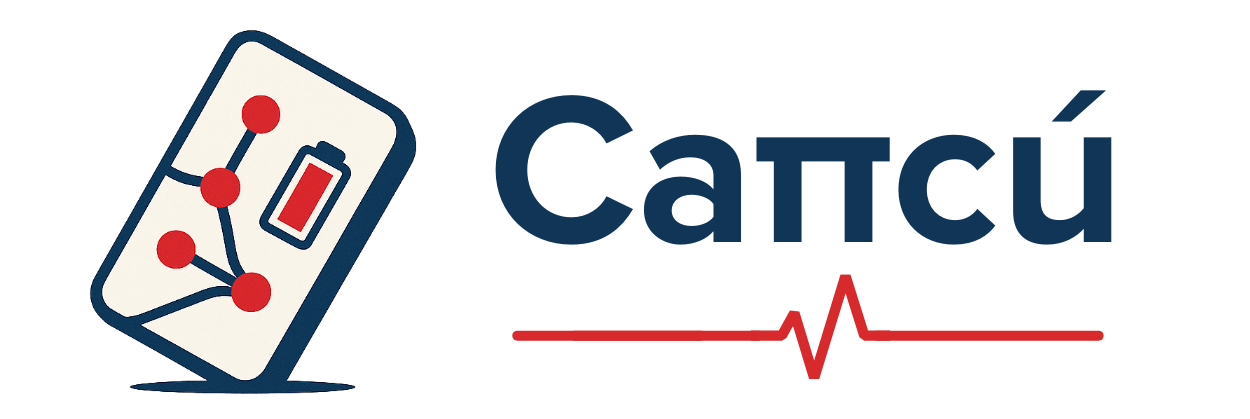}%
  \end{minipage}%
  \hfill%
  \begin{minipage}[c]{0.40\linewidth}%
    \raggedleft%
    {\footnotesize\color{CapicuMuted}\PaperVersion}%
  \end{minipage}%
  \par\vspace{4pt}%
  {\color{CapicuRule}\rule{\linewidth}{0.5pt}}%
  \par\vspace{12pt}%
  \begin{center}%
    {\sffamily\bfseries\LARGE\color{CapicuSlate}\PaperTitle}%
    \par\vspace{32pt}%
    {\normalsize\color{CapicuSlate}\PaperAuthors}%
    \par\vspace{4pt}%
    {\small\color{CapicuMuted}\PaperAffils}%
  \end{center}%
  \vspace{40pt}

  \begingroup

  \begin{tcolorbox}[colback=blue!5!white!, colframe=blue!5!white!, coltitle=white!]

    \PaperAbstractText
    \\ \\ \textbf{Publication Date:} August 1, 2026
    \\ \textbf{Correspondence:~\href{mailto:scruzromero@capicu.ai}{\texttt{scruzromero@capicu.ai}}}
  
\end{tcolorbox}
  \par\vspace{10pt}%
  \endgroup
  \par\vspace{8pt}%

}

\begin{document}

\CapicuHeader

\section{Introduction}

Edge machine learning (edge ML) has shifted from a research direction to a deployment substrate. ARM Cortex-A and Cortex-M devices, the dominant silicon in mobile SoCs, single-board computers, and embedded controllers, now rutn an increasing share of the inference workload that previously belonged to cloud GPUs. The standard pipeline trains a neural network in floating point on accelerator hardware, compresses the trained model through post-training quantization, and serves the compressed model under a CPU runtime such as ONNX Runtime, TensorFlow Lite, or PyTorch Mobile. Two architectural facts shape the engineering reality of this pipeline. First, the same compressed model is typically served across a heterogeneous fleet of edge devices spanning multiple ARM microarchitectures and silicon vendors, since edge fleets accumulate hardware over years rather than refresh in lockstep. Second, the runtime is itself a multi-backend system that selects different kernel implementations on each device through feature-detection logic at session initialization.

These two facts produce a phenomenon that has only recently been characterized rigorously. The same trained model and the same input can yield different outputs on different devices, not because of model error or training noise, but because the kernel that executes the convolution differs across devices. Floating-point non-associativity, understood since Goldberg \cite{goldberg1991every} and codified in IEEE 754 \cite{ieee754}, means that the order in which a reduction sums its operands changes the result through rounding. SIMD-width differences across hardware change the order. Schlögl et al. \cite{schlogl2023causes} gave the phenomenon a precise vocabulary on x86. Configurations of (hardware, runtime, execution provider) form \emph{equivalence classes} (EQCs) whose members produce byte-identical outputs, and \emph{remaining precision} (RP) measures pairwise drift as the number of leading identical mantissa bits between two outputs. A single TensorFlow binary running ResNet-18 on 75 x86 platforms fragments into as many as 26 EQCs, and the fragmentation tracks SSE versus AVX versus AVX-512 kernel boundaries. Subsequent work \cite{zhang2025hspi} has shown that the fragmentation is structured enough to support hardware fingerprinting from model outputs alone, with practical implications for supply-chain auditing of inference services. Adjacent results have separated other sources of inference non-determinism, including batch-invariance failures in LLM serving \cite{yuan2025nondeterminism}, GPU atomic reductions \cite{shanmugavelu2024impacts}, and cuDNN version drift on identical hardware \cite{pham2020problems}.

Two gaps in this picture are central to edge deployment. The first is the substrate gap. The EQC and RP framework was developed and validated on x86 CPUs and NVIDIA GPUs. Whether the same fragmentation holds on ARM, the dominant edge substrate, is not a corollary of the x86 result because the underlying instruction sets are structurally different. ARM and x86 implement INT8 matrix multiplication through different SIMD instructions, with different accumulator precisions and different saturation semantics. The x86 finding transfers to ARM only if the implicated mechanism transfers. The second gap is the precision gap. The x86 characterization was conducted in FP32. Edge inference is increasingly INT8, because integer arithmetic, reduced memory footprint, and dedicated dot-product instructions together yield order-of-magnitude latency improvements at acceptable accuracy cost. Post-training static quantization in the QDQ format \cite{jacob2018quantization,krishnamoorthi2018quantizing,wu2020integer} is the standard interchange representation for ONNX and TFLite INT8 deployment. Most evaluations of QDQ focus on end-task accuracy preservation rather than on whether quantization changes the numerical reproducibility picture. Practitioners routinely observe that INT8 inference is more reproducible than FP32 across hardware, but the structural reason has not been articulated in the literature.

This paper closes both gaps and identifies the boundary at which the x86 result genuinely does not extend to ARM.

\section{Related Work}

\label{sec:background}
\subsection{Floating-Point Non-Associativity and Numerical Reproducibility}
Floating-point arithmetic is not associative. Goldberg \cite{goldberg1991every} traced the consequences for scientific computing and the IEEE 754 standard \cite{ieee754} codified the rounding rules under which compliant hardware operates. Higham \cite{higham2002accuracy} catalogued the resulting reduction-order sensitivities in numerical linear algebra, and Demmel and Nguyen \cite{demmel2013fast} developed reproducible parallel summation algorithms that recover bit-exactness at the cost of additional arithmetic. The neural network inference setting inherits the same non-associativity but amplifies it through two compounding factors. Modern CPU and GPU runtimes parallelize convolutions across SIMD lanes whose width depends on the hardware, and they dispatch to different kernel implementations at runtime through feature-detection logic, so the reduction order is determined jointly by the hardware and by the runtime's dispatch policy rather than by the model author.

The cumulative effect has been documented across the inference stack. NVIDIA's cuDNN documentation \cite{cudnn_repro} states explicitly that no routines guarantee bitwise reproducibility across architectures. Pham et al. \cite{pham2020problems} measured accuracy variance from CUDA and cuDNN version changes alone on identical hardware. Shanmugavelu et al. \cite{shanmugavelu2024impacts} quantified run-to-run variability from GPU atomic reductions on Grace Hopper systems. Yuan et al. \cite{yuan2025nondeterminism} traced LLM serving non-determinism to batch-invariance failures in attention kernels rather than to floating-point non-associativity per se, and recovered deterministic inference on A100 hardware through kernel-level fixes. These results collectively establish that inference non-determinism arises from multiple distinct layers of the stack, with hardware microarchitecture being one among several sources rather than the sole source.

\subsection{Equivalent Classes and Remaining Precision Framework}
Schlögl et al. introduced the EQC and RP formalism in a sequence of papers that progressively sharpened the characterization of hardware-linked inference divergence. The forensic statement \cite{schloegl2021forensicability} demonstrated that the same trained network produces structured, hardware-dependent output deviations across CPUs and GPUs. The iNNformant work \cite{schloegl2021innformant} showed that adversarially crafted boundary samples could flip predicted labels across microarchitectures, establishing that the deviations are not merely numerical noise but can carry through to classification decisions on appropriately constructed inputs. The mechanistic statement \cite{schlogl2023causes} closed the framework: configurations of (hardware, runtime, execution provider) form equivalence classes whose members yield byte-identical outputs, and remaining precision, the count of leading identical mantissa bits between two outputs capped at 23 for IEEE binary32, quantifies pairwise drift. A single TensorFlow 2.5 binary running ResNet-18 on 75 x86 platforms fragments into as many as 26 EQCs whose boundaries align with SSE 4.2, AVX2, and AVX-512F kernel dispatch. Zhang et al. \cite{zhang2025hspi} extended the line to GPU fingerprinting through border inputs and logit distribution classifiers, identifying serving platforms with up to 100\% accuracy and demonstrating the supply-chain auditing implications of the underlying divergence.

The framework leaves two empirical questions open for edge deployment. The original characterization was performed on x86 CPUs and NVIDIA GPUs, and ARM Cortex-A devices, which dominate the edge fleet, were not covered. The original experiments used FP32 inference, and INT8 post-training quantization, which dominates edge serving, was not probed. Whether the fragmentation persists on ARM and whether quantization changes the answer are the two questions this paper addresses.

\subsection{Post-Training Quantization and the QDQ Format}
Jacob et al. \cite{jacob2018quantization} established the integer-arithmetic-only inference scheme that underlies modern INT8 deployment. Activations and weights are mapped to 8-bit integers through affine scale-and-zero-point quantizers, and convolutions reduce to INT8 dot products accumulated into wider integer destinations. Krishnamoorthi \cite{krishnamoorthi2018quantizing} consolidated the practitioner guidance for post-training and quantization-aware training of CNNs. Wu et al. \cite{wu2020integer} reported empirical accuracy preservation across image and language workloads. Nagel et al. \cite{nagel2020up} refined the rounding policy at quantization time through AdaRound. The QDQ format wraps an FP32 graph with explicit QuantizeLinear and DequantizeLinear nodes at every layer boundary, exposing the quantization parameters as graph metadata that the runtime's optimizer can pattern-match against. This representation is now the standard interchange format for INT8 ONNX models and is the format used by TensorFlow Lite's full-integer quantization path.

The accumulator design inside a quantized convolution is the variable that matters for cross-platform reproducibility, and the literature is uneven on this point. Colbert et al. \cite{colbert2023a2q} derived bit-width bounds that guarantee overflow avoidance under $L_{1}$-norm weight constraints, establishing the formal precondition under which integer accumulation is exact. The oneDNN documentation \cite{onednn_int8} describes how the x86 VPMADDUBSW instruction introduces saturating INT16 intermediates between INT8 multiplicands and the INT32 destination, and documents a weight-scaling workaround to mitigate the resulting accuracy degradation on AVX2 platforms without VNNI. The ONNX Runtime quantization documentation \cite{onnxrt_quant} similarly warns of saturation on x86 without VNNI and recommends a reduce-range strategy. These resources treat the INT16 saturation step as a single-platform accuracy concern. None of them connect it to cross-platform EQC divergence, which is the connection this paper draws.

\subsection{ARM SIMD and Runtime Microkernel Dispatch}
The two ARM microarchitectures most relevant to current edge deployment differ in a specific INT8 capability. The Cortex-A72 implements ARMv8-A with NEON, providing 128-bit SIMD over INT8 lanes but no native INT8 dot product. The Cortex-A76 implements ARMv8.2-A with the dot-product extension, exposing the SDOT instruction that computes four INT8 multiply-accumulates into INT32 in one cycle. Gope et al. \cite{gope2025arm} describe Arm's reference INT8 kernels for Cortex-A and document the throughput advantage of the SDOT path. San Adri'an et al. \cite{san_adrian2025cambrian} survey the broader ecosystem of mixed-precision matrix multiplication microkernels across ARM NEON, NEON with dot-product, SVE2, and RISC-V vector extensions, characterizing the performance landscape but not output reproducibility.

XNNPACK \cite{dukhan2021xnnpack,dukhan2019indirect} is the CPU delegate that backs TensorFlow Lite and serves as an execution provider under ONNX Runtime. It detects available SIMD features at initialization through the \texttt{cpuinfo} library and dispatches to one of several microkernel families per operator. On ARM, the relevant fork is between \texttt{qs8-gemm} kernels that compile around SDOT when \texttt{asimddp} is reported and fallback kernels that use \texttt{SMULL}, \texttt{SMLAL}, and \texttt{SADALP} when it is not. The mathematical operation specified by an ONNX or TFLite op is identical across the two paths, but the implementations differ in accumulation order, fused multiply-add usage, and SIMD width. All are IEEE-conforming and none are bit-equivalent by construction, which is the precise condition under which the EQC framework predicts fragmentation.

\subsection{Edge Inference Reproducibility Across Hardware}
Empirical studies of edge inference across heterogeneous ARM hardware have focused on performance rather than reproducibility. Jain et al. \cite{jain2020efficient} benchmarked INT8 across PyTorch-FBGEMM, QNNPACK, and TensorFlow Lite on Intel Cascade Lake and Raspberry Pi 4B, measuring latency and throughput. The MLPerf Inference benchmark suite \cite{reddi2020mlperf} provides standardized cross-platform latency and accuracy measurements but does not report whether the outputs are byte-identical across platforms. The San Adri'an et al. survey compared microkernel performance across ARM and RISC-V platforms including a Cortex-A72 Pi, again on performance metrics. To our knowledge, no prior work has directly characterized bit-exact INT8 inference outputs across ARM microarchitectures under runtimes with confirmed different microkernel dispatch. This paper supplies that characterization, identifies its structural cause, and locates the boundary at which the same characterization would fail under different runtime arithmetic.

\section{Methods}
\label{sec:setup}

\subsection{Hardware}

We instrument four Raspberry Pi nodes spanning three ARM microarchitectures, with two A53 nodes on different SoCs to control for silicon differences below the microarchitecture level (Table~\ref{tab:hw}). The four-device set is used for the small-CNN ONNX Runtime experiments. Production-CNN TFLite experiments use Node A and Node B because the smaller A53 boards do not provide enough RAM or compute headroom for ResNet50V2.

\begin{table}[!t]
\caption{Hardware Configurations\label{tab:hw}}
\centering
\footnotesize
\begin{tabular}{@{}llllll@{}}
\toprule
Device & Board & SoC & CPU & ISA & RAM \\
\midrule
Node A & Pi 5      & BCM2712 & Cortex-A76 & ARMv8.2-A & 4 GB \\
Node B & Pi 4B     & BCM2711 & Cortex-A72 & ARMv8-A   & 1.8 GB \\
Node C & Pi 3B     & BCM2837 & Cortex-A53 & ARMv8-A   & 1 GB \\
Node D & Pi Zero 2W & RP3A0  & Cortex-A53 & ARMv8-A   & 512 MB \\
\bottomrule
\end{tabular}
\\[2pt]
\raggedright
Nodes C and D share the Cortex-A53 microarchitecture on different SoCs. Node A is the only device exposing \texttt{asimddp}, which enables SDOT microkernel dispatch under XNNPACK.
\end{table}

\subsection{Models, Data, and Protocol}

\textsc{CifarSmall} is a 12-layer CNN (Conv, ReLU, MaxPool, Conv, ReLU, MaxPool, Flatten, Gemm, ReLU, Gemm, ReLU, Gemm) trained for 30 epochs on CIFAR-10 \cite{krizhevsky2009cifar} with seed 0 to test accuracy 63.4\%. It is exported in two formats, FP32 and INT8 QDQ (per-tensor symmetric quantization, training-set calibration), and serves as the diagnostic model. MobileNetV2 \cite{sandler2018mobilenetv2} (3.4M parameters, max channel $K = 1280$) and ResNet50V2 \cite{he2016identity} (25.6M parameters, $K = 2048$) serve as production CNNs. Both are exported to TFLite FP32 and quantized via full-integer post-training quantization with 200 calibration images.

Inputs for the small-CNN experiments are the first 1{,}000 CIFAR-10 test images, normalized to $[0,1]$, NCHW float32. Inputs for the production-CNN experiments are 500 ImageNet ILSVRC2012 validation images selected lexicographically, preprocessed once on a development machine (resize to 256, center-crop to 224, normalize to $[0,1]$). All inputs are SHA-256-verified byte-identical across nodes. Single-image inference, single-threaded (\texttt{num\_threads=1}), three warmup inferences discarded before measurement. Output hash: SHA-256 of the float32 output tensor in row-major byte order.

\subsection{Conditions and Runtime Stacks}

The small-CNN experiment uses five \ort{} conditions (Table~\ref{tab:conds}). The Debian-packaged \texttt{python3-onnxruntime} 1.21 (on Node A only) is the sole build we identified with a working \texttt{XnnpackExecutionProvider} on ARM. It requires ARMv8.2-A and raises \texttt{SIGILL} at import on Nodes B through D. The PyPI \texttt{onnxruntime} 1.26 wheel, used on Nodes B through D, provides CPU EP only in practice. The production-CNN experiment uses three TFLite conditions on the same \texttt{ai-edge-litert==2.1.4} aarch64 wheel from PyPI: Condition A (INT8, A76, XNNPACK selects SDOT microkernels), Condition B (FP32, A76, cross-precision reference), and Condition C (INT8, A72, XNNPACK selects NEON-only microkernels).

\begin{table}[!t]
\caption{Small-CNN ONNX Runtime Conditions\label{tab:conds}}
\centering
\footnotesize
\begin{tabular}{@{}lllll@{}}
\toprule
Cond & Device & EP (intended) & EP (resolved) & ORT \\
\midrule
A & Node A & XNNPACK & XnnpackExecutionProvider & 1.21 \\
B & Node A & CPU     & CPUExecutionProvider     & 1.21 \\
C & Node B & CPU     & CPUExecutionProvider     & 1.26.0 \\
D & Node C & CPU     & CPUExecutionProvider     & 1.26.0 \\
E & Node D & CPU     & CPUExecutionProvider     & 1.26.0 \\
\bottomrule
\end{tabular}
\\[2pt]
\raggedright
``EP (intended)'' is the provider requested at \texttt{InferenceSession} creation. ``EP (resolved)'' is the provider actually used at runtime, parsed from per-run telemetry rather than the session's advertised provider list.
\end{table}

\subsection{Silent XNNPACK Fallback and Methodological Implication}

We initially attempted to run MobileNetV2 and a ResNet variant on Node A using pip \ort{} 1.25 with \texttt{XnnpackExecutionProvider}. The provider was accepted at \texttt{InferenceSession} creation without exception, but every recorded \texttt{backend\_providers} field resolved to \texttt{CPUExecutionProvider}. The runs were a CPU-versus-CPU comparison and are excluded from positive evidence. The artifact bounds the dispatch positive in Section~\ref{sec:dispatch} to \textsc{CifarSmall} on Cortex-A76 under the Debian build. It is also a concrete instance of what Sculley et al. \cite{sculley2015hidden} call configuration debt: dependency state diverges silently from the developer's mental model. A practitioner reading the session's advertised provider list would conclude XNNPACK is active when it is not, with downstream consequences for any analysis that treats execution provider as an experimental factor. We recommend reading the resolved provider from per-run runtime telemetry rather than from the session object.

The production-CNN experiment uses TFLite specifically to sidestep this artifact, since TFLite/XNNPACK dispatch behavior is documented and additionally verifiable through timing signatures (Section~\ref{sec:generalization}). Two preliminary framework experiments at the same scale, ONNX Runtime 1.25.0 \texttt{CPUExecutionProvider} and PyTorch 2.11.0 qnnpack INT8, also produced byte-identical outputs on both nodes for MobileNetV2 and ResNet18. Post-hoc analysis showed the ONNX Runtime aarch64 wheel uses reference kernels without architecture-specific dispatch, so that null carries no positive information about microkernel divergence. The qnnpack dispatch claim was not profiled. The TFLite experiment resolves the ambiguity through documented dispatch logic and corroborating timing evidence.

\section{Discussion of Results/Findings}
\subsection{The Dispatch Axis on ARM}
\label{sec:dispatch}

Under CPU EP (Conditions B, C, D, E), all 1{,}000 outputs are byte-identical across all four nodes for both small-CNN on fp32 and int8 precision ($\rpn = 23$ for every image, one EQC, 0 of 1{,}000 divergent). This holds across three boundaries simultaneously: A53 versus A72 versus A76 (three microarchitectures), BCM2837-A53 versus RP3A0-A53 (two SoCs at fixed microarchitecture), and Debian \ort{} 1.21 versus pip \ort{} 1.26 (two runtime builds). For the operators in this model under the CPU provider, ARM microarchitecture does not fragment EQCs.

The picture inverts when only the execution provider changes. For fp32 small-CNN, Condition A (XNNPACK on A76) and Condition B (CPU on A76) differ on 1{,}000 of 1{,}000 images (100\% divergence). Mean RP is 14.97 with standard deviation 2.28 and range 3 to 20, meaning 8.03 mantissa bits are lost per image on average. The RP distribution is a broad ridge several bits wide (Fig.~\ref{fig:rp-distributions}), not a few outliers around the reference, which indicates a systematic dispatch effect rather than a tail event.

INT8 QDQ collapses both axes at once. For int8 small-CNN, all five conditions are byte-identical ($\rpn = 23$ for every image, one EQC, 0 of 1{,}000 divergent). The reproducibility gain spans hardware (A53, A72, A76 across two SoC families) and dispatch (XNNPACK and CPU on Cortex-A76) simultaneously. Table~\ref{tab:summary} summarizes per-condition divergence and RP relative to Condition B.

\begin{table}[!t]
\caption{Per-Condition Divergence and RP vs. Condition B (FP32, CPU, A76)\label{tab:summary}}
\centering
\footnotesize
\begin{tabular}{@{}llrrrr@{}}
\toprule
Model & Cond & Node & Diverge & Mean RP & Range \\
\midrule
fp32 & A    & Node A & 100.0\% & 14.97 & 3 to 20 \\
fp32 & B    & Node A & 0.0\%   & 23.00 & 23 \\
fp32 & C    & Node B & 0.0\%   & 23.00 & 23 \\
fp32 & D    & Node C & 0.0\%   & 23.00 & 23 \\
fp32 & E    & Node D & 0.0\%   & 23.00 & 23 \\
int8 & A--E & all    & 0.0\%   & 23.00 & 23 \\
\bottomrule
\end{tabular}
\\[2pt]
\raggedright
$n = 1{,}000$ images per FP32 condition. Combined $n = 5{,}000$ for INT8.
\end{table}

\begin{figure}[!t]
\centering
\includegraphics[width=.75\linewidth]{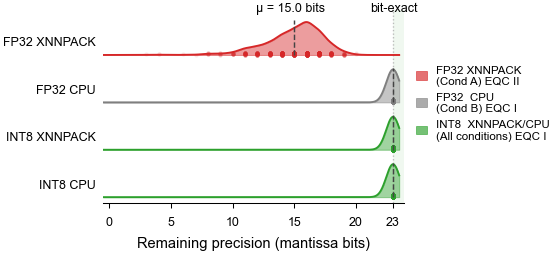}
\caption{Distribution of remaining precision over 1{,}000 CIFAR-10 test images for the four distinct condition-and-format combinations on Cortex-A76. Only the FP32 XNNPACK condition departs from byte-identity with the reference. The other three collapse to the $\rpn = 23$ ceiling. Conditions C, D, and E under CPU EP coincide with Condition B and are summarized in Table~\ref{tab:summary} rather than shown as separate ridges.}
\label{fig:rp-distributions}
\end{figure}

The structure of the FP32 case presents a dispatch fingerprint invisible to label-level accuracy metrics: 100\% of output hashes disagree while 0\% of predicted labels do for this model on these images. Anything that consumes continuous outputs (confidence calibrators \cite{guo2017calibration}, anomaly scorers, audit logs, per-device A/B comparisons) inherits the drift. Whether that matters depends on what is downstream of the model, not on whether the model classifies correctly.

\subsection{The H1+H2 Mechanism}
\label{sec:mechanism}

Why does INT8 QDQ collapse the dispatch fingerprint? Three sub-experiments on Node A (XNNPACK versus CPU, \ort{} 1.21, $n = 50$ CIFAR-10 images per scenario) isolate the answer and motivate a two-part structural invariant.

\subsection{Localization to Conv0}

We extract single-output prefix subgraphs of fp32 small-CNN at each of the 12 layer boundaries using \texttt{onnx.utils.extract\_model} and compare execution providers (Table~\ref{tab:layerwise}, FP32 column). Divergence originates at Conv0, where mean RP is only 2.56 bits, and partially attenuates as the signal moves downstream, reaching mean RP 15.32 at the output. The profile is inconsistent with an error-accumulation story (in which downstream operators would amplify drift) and consistent with low-order mantissa noise being smoothed by subsequent reductions, including ReLU clipping and MaxPool selection.

We then insert one \texttt{QuantizeLinear} and \texttt{DequantizeLinear} pair only at the input, using $s = 0.003922$ and $z = -128$ from the calibrated INT8 model, leaving all internal computation in FP32 (Table~\ref{tab:layerwise}, input-only QDQ column). Conv0 becomes bit-exact across execution providers, which establishes the first half of the mechanism.

\begin{cdef}[H1, discrete-grid absorption]
\label{def:h1}
Let $C$ be a Conv or Gemm operator in an ONNX graph whose input activation tensor $a$ satisfies $a_i \in \grid_{s,z}$ elementwise for some scale $s > 0$ and zero-point $z \in \mathbb{Z}$. We call $C$ \emph{discrete-grid absorbing} if the FP32 output of $C$ is byte-identical under the execution providers being compared.
\end{cdef}

Conv1 immediately reverts to divergence because Conv0's continuous-valued FP32 outputs are no longer on the grid. Restoring grid membership at every layer boundary is therefore necessary, not just sufficient, to extend H1 through the network.

\subsection{The Requantization Ratchet}

We instrument the full int8 small-CNN graph with \texttt{Identity} nodes that expose each of the 9 \texttt{QuantizeLinear} activation outputs as named tensors and compare them element-wise between execution providers. All nine INT8 tensors are byte-identical (0 of 50 divergent at every hook), which is the second structural property.

\begin{cdef}[H2, requantization ratchet]
\label{def:h2}
A QDQ-quantized graph satisfies the \emph{requantization ratchet} property if, for every \texttt{QuantizeLinear} node in the graph, the INT8 tensor it emits is byte-identical under the execution providers being compared.
\end{cdef}

H1 makes any single Conv deterministic across dispatch paths when its inputs lie on the int8 grid. H2 ensures that the precondition is restored at every layer boundary. The combined invariant follows.

\begin{cprop}[H1+H2, empirical statement]
\label{prop:h1h2}
Let $G$ be a QDQ-quantized graph composed of the operators Conv, Gemm, ReLU, MaxPool, Flatten, QuantizeLinear, and DequantizeLinear. If every Conv and Gemm in $G$ is preceded by a DequantizeLinear whose input lies on $\grid_{s,z}$ for the calibrated $(s, z)$, and every activation tensor in $G$ is followed by a QuantizeLinear before the next Conv or Gemm, then $G$'s FP32 output is byte-identical under the execution providers being compared.
\end{cprop}

We state Proposition~\ref{prop:h1h2} as an empirical property on the operators present in \textsc{CifarSmall} under the specified runtime build. The most plausible proximate mechanism is that \ort{}'s graph optimizer identifies the QDQ pattern and routes Conv and Gemm through \texttt{QLinearConv} and \texttt{QGemm} kernels that perform the dot product in an INT32 accumulator over INT8 multiplicands. Integer addition is associative, so reduction order is irrelevant and SIMD-width differences between XNNPACK and CPU kernels become invisible at the output. A formal treatment would trace dispatch logic operator by operator and verify that the grid-snapped-input precondition triggers the integer kernel path in both providers. Section~\ref{sec:generalization} probes the prediction directly: under a different runtime where different microkernel families are confirmable on different devices, H1+H2 says outputs should remain byte-identical.

\begin{table}[!t]
\caption{Layer-Wise Mean RP, XNNPACK vs.\ CPU on Cortex-A76\label{tab:layerwise}}
\centering
\footnotesize
\begin{tabular}{@{}llrrr@{}}
\toprule
Layer & Op & FP32 & Input-only QDQ & Full INT8 QDQ \\
\midrule
0  & Conv     & 2.56  & 23.00 & 23.00 \\
1  & ReLU     & 3.80  & 23.00 & 23.00 \\
2  & MaxPool  & 5.50  & 23.00 & 23.00 \\
3  & Conv     & 7.48  & 8.70  & 23.00 \\
4  & ReLU     & 8.34  & 10.04 & 23.00 \\
5  & MaxPool  & 10.92 & 12.16 & 23.00 \\
6  & Flatten  & 10.92 & 12.16 & 23.00 \\
7  & Gemm     & 11.22 & 13.48 & 23.00 \\
8  & ReLU     & 12.92 & 14.22 & 23.00 \\
9  & Gemm     & 10.50 & 14.66 & 23.00 \\
10 & ReLU     & 11.64 & 15.10 & 23.00 \\
11 & Output   & 15.32 & 15.48 & 23.00 \\
\bottomrule
\end{tabular}
\\[2pt]
\raggedright
FP32: divergence originates at Conv0 and attenuates downstream. Input-only QDQ: H1 makes Conv0 bit-exact but Conv1 reverts because the grid precondition is lost. Full INT8 QDQ: H2 preserves the precondition through the network.
\end{table}

\begin{figure}[!t]
\centering
\includegraphics[width=.65\linewidth]{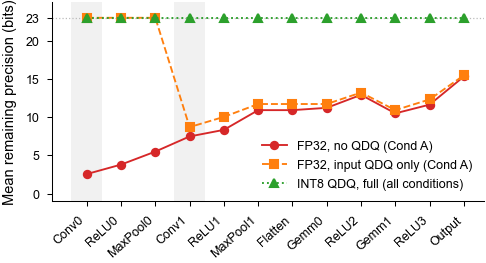}
\caption{Per-layer mean RP across the three quantization scenarios in Table~\ref{tab:layerwise}, measured between XNNPACK and CPU on Cortex-A76. FP32 with no QDQ starts at Conv0 with $\rpn = 2.56$ and partially attenuates downstream. FP32 with input-only QDQ makes Conv0 bit-exact but Conv1 immediately diverges. Full INT8 QDQ is bit-exact at every layer boundary.}
\label{fig:layerwise-rp}
\end{figure}

\section{Conclusions and Recommendations/Future Directions}
\subsection{Generalization to Production CNNs}
\label{sec:generalization}

The small-CNN result and H1+H2 together predict that bit-exact output agreement should extend to deeper CNNs under runtimes where different microkernel families are confirmably executing on different ARM devices, provided the channel dimension fits the INT32 accumulator. We test this on MobileNetV2 and ResNet50V2 under TFLite with XNNPACK on Cortex-A76 (Condition A, SDOT) and Cortex-A72 (Condition C, NEON multiply-accumulate).

A null result at the output level is only interpretable if the hardware actually exercised different arithmetic paths. We cannot claim that different microkernels produced the same output if both devices ran the same code. Two timing signals independently support different dispatch (Fig.~\ref{fig:timing}). First, the INT8/FP32 speedup on A76 is $5.04\times$ for ResNet50V2, substantially exceeding the $4\times$ memory-bandwidth reduction from INT8 weight compression. The additional throughput is consistent with SDOT computing four INT8 multiply-accumulates per instruction cycle. Second, the same INT8/FP32 comparison on A72 yields approximately $1.0\times$, meaning INT8 and FP32 consume nearly identical time. Without SDOT, the NEON INT8 path provides no compute-throughput advantage over FP32 for these layer shapes at batch size 1. This asymmetric signature, large INT8 speedup on A76 and negligible speedup on A72, is not explained by clock speed, cache hierarchy, or memory bandwidth alone. It is consistent with the documented XNNPACK dispatch rule that assigns \texttt{neondot} kernels when \texttt{cpuinfo} reports \texttt{asimddp}, which A76 does and A72 does not. Direct instruction-level profiling via \texttt{perf stat} would provide definitive confirmation. We treat documented source behavior plus matching timing signatures as strong but indirect evidence.

\begin{figure*}[!t]
\centering
\includegraphics[width=\linewidth]{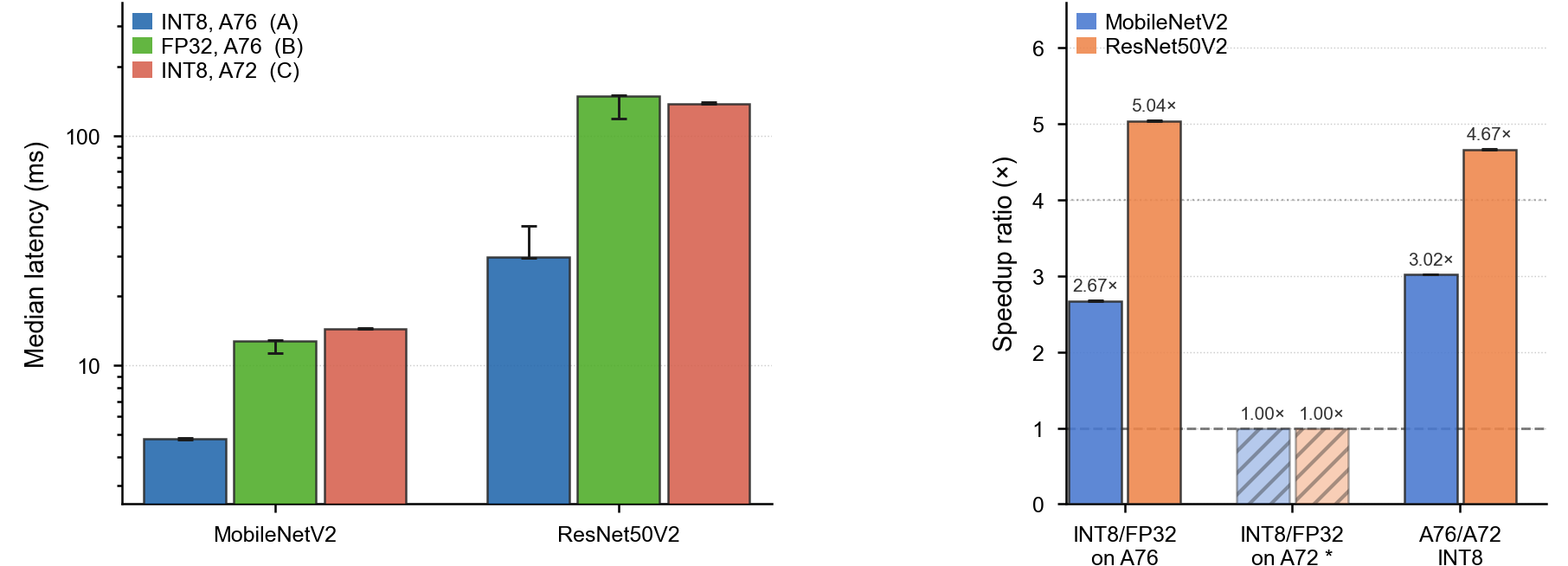}
\caption{Timing evidence for different XNNPACK microkernel dispatch on production CNNs. INT8 inference on A76 (Condition A) is $5.04\times$ faster than FP32 (Condition B) for ResNet50V2, exceeding the $4\times$ speedup from weight compression alone. On A72 (Condition C), INT8 and FP32 latency are nearly identical (137.80 ms versus approximately 138 ms). The asymmetric signature is consistent with SDOT throughput on A76 and its absence on A72.}
\label{fig:timing}
\end{figure*}

\subsection{Byte-Identity at Every Intermediate Tensor}

Across both models and 500 ImageNet images per model, every final output between A76 (SDOT) and A72 (NEON) is byte-identical. A layer-by-layer view rules out the alternative explanation that intermediate divergence exists but is absorbed by the requantization step between layers. On a 50-image subset, we captured all intermediate tensors by querying the TFLite interpreter after each \texttt{invoke()}. For MobileNetV2, all 109 tensors (53 INT8, 54 INT32, 2 FP32) are byte-identical between A and C across all 50 images. For ResNet50V2, all 156 tensors (97 INT8, 57 INT32, 2 FP32) are byte-identical. The raw INT32 bias-add accumulator values, namely the integer dot products before dequantization, are themselves identical. This is the direct empirical signature predicted by H1+H2: not requantization absorbing divergence, but no divergence at any stage.

Top-1 classification agreement between A and C is 500/500 on both models. INT8 versus FP32 top-1 disagreement is 40/500 (MNV2) and 34/500 (RN50V2), reflecting quantization fidelity. Every such disagreement is identical on both platforms, meaning no top-1 disagreement is attributable to the hardware axis.


\subsection{An Equivalence Proof for the Two Kernel Paths}

The byte-identity is not coincidental. The SDOT instruction (\texttt{vdotq\_lane\_s32}) computes
\begin{equation}
\texttt{acc} \mathrel{+}= a_0 b_0 + a_1 b_1 + a_2 b_2 + a_3 b_3,
\label{eq:sdot}
\end{equation}
with $a_i, b_i \in [-128,127]$ and $\texttt{acc} \in \mathbb{Z}_{32}$. Four products are accumulated into a single INT32 register per instruction. Per the ARM Architecture Reference Manual, no saturation is applied. The instruction name carries no \texttt{SQ} or \texttt{UQ} prefix and the pseudocode contains no \texttt{SignedSat} call. Overflow wraps modulo $2^{32}$. The NEON fallback path widens INT8 operands transiently to INT16 via \texttt{SMULL} and \texttt{SMLAL}, then accumulates into INT32 via \texttt{SADALP}. The INT16 values are intermediate operands, not a running sum. In the fallback kernel \texttt{c8-neon-mull.c.in}, a named \texttt{int16x8\_t vprod} holds exactly two INT8 by INT8 products before widening to INT32.

\begin{cprop}
\label{prop:equiv}
For any $K$ pairs $(a_k, b_k) \in [-128,127]^2$ with $K \leq \lfloor 2^{31} / 16{,}384 \rfloor = 131{,}072$, the SDOT and NEON multiply-accumulate paths produce the same INT32 result.
\end{cprop}
\begin{proof}
Each product satisfies $|a_k b_k| \leq 127 \times 128 = 16{,}256 < 2^{14}$. The sum of $K$ such products is bounded by $K \cdot 16{,}256$. For $K \leq 131{,}072$ this is bounded by $2^{31} - 1$, so no INT32 overflow occurs. Both paths compute $\sum_{k=0}^{K-1} a_k b_k$ exactly in INT32. Since integer addition is commutative and associative and all intermediate operations are exact (no rounding, no saturation at any step in either path), the grouping of four products per SDOT instruction versus one product per NEON iteration does not affect the result.
\end{proof}

The channel dimensions of MobileNetV2 ($K_{\max} = 1280$) and ResNet50V2 ($K_{\max} = 2048$) satisfy the bound by orders of magnitude, as does every standard CNN architecture in the edge deployment literature. XNNPACK additionally constrains symmetric-quantized weights to $[-127, 127]$ \cite{dukhan2021xnnpack}, which ensures the \texttt{SMLAL} INT16 intermediate products fit in 15 bits and the two-product partial sum in the fallback kernel does not saturate INT16. Proposition~\ref{prop:equiv} is the operator-level expression of H1+H2 under the specific XNNPACK kernel implementations involved.

\subsection{The x86 Boundary}
\label{sec:x86}

If integer accumulation in INT32 is associative and ARM kernels respect that, why does x86 diverge on the same INT8 inference task? The answer is a specific x86 instruction, \texttt{PMADDUBSW}, which has no ARM analogue. The x86 TFLite INT8 GEMM path uses
\begin{equation}
\text{result}_i = \text{saturate}_{\text{int16}}(a_{2i} \cdot b_{2i} + a_{2i+1} \cdot b_{2i+1}),
\end{equation}
which clips to $[-32{,}768, 32{,}767]$. For unsigned activations with zero-point shifted into $[0,255]$, saturation triggers for moderate values. As an example, $200 \times 90 + 210 \times 90 = 36{,}900 > 32{,}767$. Different SSE and AVX implementations partition the reduction tree differently, determining which operand pairs are processed together by \texttt{PMADDUBSW} and therefore which pairs saturate. The integer kernel itself is the source of order-dependence (Fig.~\ref{fig:accum}). This is the root cause of x86 EQC divergence on INT8 inference, not floating-point non-associativity but integer saturation at an intermediate width. The oneDNN documentation \cite{onednn_int8} describes the mechanism explicitly and documents the $0.5\times$ weight-scaling workaround for non-VNNI platforms. Intel's VNNI extension (\texttt{VPDPBUSD}) resolves the issue by accumulating directly into INT32 \cite{onednn_int8}, mirroring the ARM design.

ARM has no \texttt{PMADDUBSW} equivalent. Neither SDOT nor the NEON \texttt{SMLAL} and \texttt{SADALP} path produces saturating intermediates. The ARM null result is therefore not a contradiction of Schl\"ogl et al. but a delineation. Their finding applies to x86 implementations that use saturating INT16 intermediates. It does not apply to ARM XNNPACK, where the saturation mechanism does not exist.

A practical consequence follows for hardware fingerprinting. The HSPI framework \cite{zhang2025hspi} exploits hardware-linked numerical divergence to identify serving platforms from model outputs. The analogous approach for ARM, namely using border inputs or logit distribution classifiers to distinguish A76 from A72 based on INT8 CNN outputs, cannot succeed for XNNPACK-dispatched inference because outputs are byte-identical for any input satisfying the channel bound. Hardware auditing methods for ARM edge deployments must target other layers of the stack, such as FP32 inference, autoregressive generation, or timing side channels.

\begin{figure*}[!t]
\centering
\includegraphics[width=\linewidth]{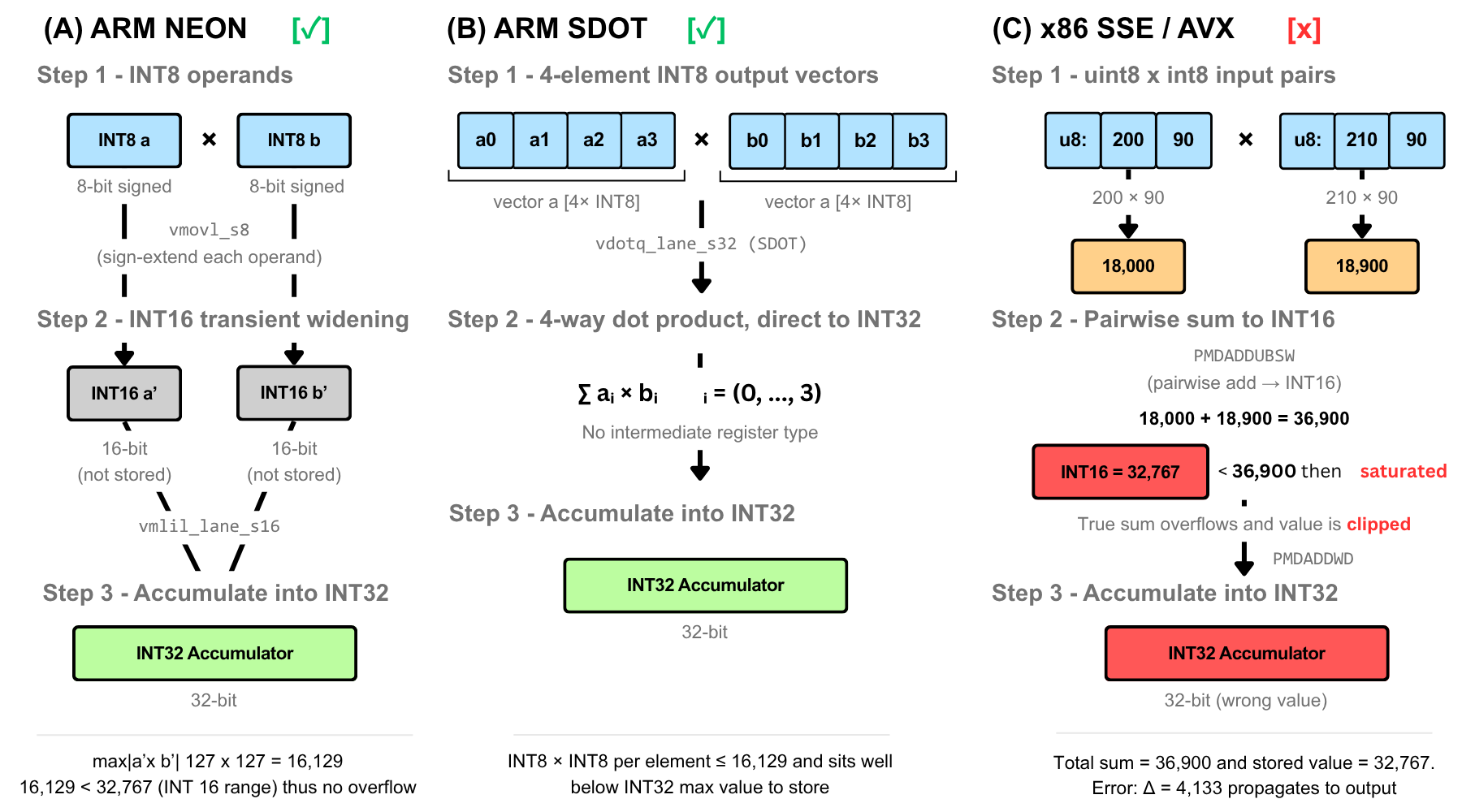}
\caption{INT8 GEMM accumulation paths on ARM and x86. ARM NEON widens INT8 operands transiently to INT16 via \texttt{SMULL} and \texttt{SMLAL}, then accumulates into INT32 via \texttt{SADALP} with no saturation at any step. ARM SDOT computes four INT8 products and accumulates into INT32 in one cycle with no saturation. The x86 \texttt{PMADDUBSW} instruction produces saturating INT16 output when the sum of two uint8 by int8 products exceeds 32{,}767, reachable for plausible quantized activations. Different SSE and AVX implementations partition the reduction tree differently, causing saturation to occur at different operand pairs and producing different INT32 accumulator values \cite{onednn_int8}.}
\label{fig:accum}
\end{figure*}

\subsection{Deployment Consequences, Limitations, and Conclusion}

\paragraph{Latency and Label Stability}

Across the small CNN, INT8 reduces mean per-image inference latency by 1.29 to 1.46 times across all five conditions, with the largest gain on A76 under CPU EP (Table~\ref{tab:latency}). The speedup holds across XNNPACK and CPU providers and across three microarchitectures, suggesting it comes from intrinsic 8-bit arithmetic and reduced memory traffic rather than from a single runtime-tuned kernel. On the production CNNs under TFLite, INT8 reduces median per-image latency from 12.82 to 4.80 ms on MobileNetV2 ($2.67\times$) and from 148.91 to 29.54 ms on ResNet50V2 ($5.04\times$) on Cortex-A76, with the ResNet50V2 figure attributable to SDOT throughput.

On the small CNN, INT8 changes the predicted class on 27 of 1{,}000 images relative to the FP32 reference, a rate of 2.7\%, concentrated in the model's low-confidence region (mean FP32 confidence 0.42 on the disagreement subset versus 0.83 across all 1{,}000 images). A 0.5 confidence threshold on FP32 outputs recovers full label consistency on this dataset. On MobileNetV2 and ResNet50V2 under TFLite, INT8 versus FP32 top-1 disagreement is 8.0\% and 6.8\% respectively. Every such disagreement is identical on A76 and A72, so the hardware contribution to label drift is zero.

For deployment, the operational consequence is direct. Continuous outputs are reproducible across the ARM fleet under INT8 QDQ. A small and predictable subset of labels changes against the FP32 reference, concentrated in the low-confidence tail. Which tradeoff a deployment should prefer depends on what consumes the model's outputs, not on which mode is correct.

\begin{table}[!t]
\caption{Mean Per-Image Inference Latency, Small CNN, 1{,}000 CIFAR-10 Images\label{tab:latency}}
\centering
\footnotesize
\begin{tabular}{@{}lllrrr@{}}
\toprule
Cond & $\mu$arch & EP & FP32 (ms) & INT8 (ms) & Speedup \\
\midrule
A & A76 & XNNPACK & 0.194 & 0.143 & 1.36 \\
B & A76 & CPU     & 0.288 & 0.197 & 1.46 \\
C & A72 & CPU     & 0.895 & 0.678 & 1.32 \\
D & A53 & CPU     & 2.103 & 1.626 & 1.29 \\
E & A53 & CPU     & 2.385 & 1.806 & 1.32 \\
\bottomrule
\end{tabular}
\end{table}

\paragraph{Limitations}

The dispatch positive in Section~\ref{sec:dispatch} is measured on one small CNN on Cortex-A76 because that is the only Cortex device for which we obtained a working XNNPACK execution provider under \ort{} on ARM, given the silent fallback in pip wheels. The H1+H2 mechanism in Section~\ref{sec:mechanism} is empirically demonstrated on the operators present in \textsc{CifarSmall}: Conv, Gemm, ReLU, MaxPool, Flatten, QuantizeLinear, DequantizeLinear. Proposition~\ref{prop:h1h2} is stated as an empirical property, not as a theorem. A formal treatment would trace \ort{}'s dispatch logic to confirm that the QDQ pattern routes Conv and Gemm through integer kernels in both providers. The generalization in Section~\ref{sec:generalization} is performed on two ARM microarchitectures. The equivalence proof generalizes to any platform running XNNPACK NEON or dot-product kernels with INT32 accumulators, but empirical confirmation on Cortex-A55, A78, the Cortex-X series, and Qualcomm Oryon is reserved for follow-up. Float32 inference was not studied for cross-hardware divergence on the production CNNs. Floating-point non-associativity may produce divergence even on ARM. The thermal state of Pi 5 (88 to 91 degrees Celsius during sustained ResNet50V2 inference) introduces latency variability that does not affect output correctness but limits the precision of timing-based dispatch claims. Direct instruction-count confirmation via \texttt{perf stat} has not been performed. Models with channel dimensions $K > 131{,}072$ or asymmetric weight ranges exceeding $[-127, 127]$ fall outside the scope of Proposition~\ref{prop:equiv}.

\section{Conclusion}

On ARM edge hardware, microarchitecture is not observable in the outputs of a fixed neural network under default CPU dispatch, but the execution provider is. INT8 QDQ post-training quantization collapses both axes to a single equivalence class as a structural property of the QDQ encoding rather than as a consequence of runtime tuning. The structural account predicts and the production-CNN measurement confirms that bit-exact agreement persists under runtimes where different ARM microkernel families demonstrably execute, with the predicted INT32-accumulator equivalence holding for every layer of MobileNetV2 and ResNet50V2 between Cortex-A76 SDOT and Cortex-A72 NEON multiply-accumulate paths. The boundary of the phenomenon is the x86 \texttt{PMADDUBSW} instruction. For practitioners deploying quantized CNNs across heterogeneous ARM fleets, INT8 is the reproducible mode, the relevant behavioral variation axis is precision and not microarchitecture, and any hardware-auditing method that relies on output divergence to fingerprint ARM edge devices must target some layer of the stack other than INT8 CNN inference under XNNPACK.

\section*{Acknowledgments}

\bibliography{preprint-int8-quant-invariant-arm/ref}

\end{document}